\documentclass[final]{siamart1116}

\usepackage{amsmath,amssymb, amsfonts}
\usepackage{algorithm}
\usepackage[mathscr]{euscript}
\usepackage{algpseudocode}
\usepackage{bm}
\usepackage{graphicx}
\usepackage{booktabs} 
\usepackage{enumerate}
\newcommand{\ra}[1]{\renewcommand{\arraystretch}{#1}}
\usepackage{float}
\usepackage{mathtools}
\usepackage{subcaption}
 \usepackage{relsize}

\title{An Approximation Theoretic Perspective of Sobol' Indices with Dependent Variables \thanks{This work was partially supported by the National Science Foundation under grant DMS-1522765, and grant DMS-1638521 to the Statistical and Applied Mathematical Sciences Institute. Any opinions, findings, and conclusions or recommendations expressed in this material are those of the authors and do not necessarily reflect the views of the National Science Foundation.}}
\author{J.~L. Hart\thanks{Department of Mathematics, North Carolina State University, Raleigh, NC 27695-8205 \email{jlhart3@ncsu.edu}} \and P.~A. Gremaud\thanks{Department of Mathematics, North Carolina State University, Raleigh, NC 27695-8205 \email{gremaud@ncsu.edu}}}

\newcommand{\R}{\mathbb{R}}

\newcommand{\x}{\mathbf x}

\newcommand{\y}{\mathbf y}

\newcommand{\V}{\operatorname{Var}}

\newcommand{\C}{\operatorname{Cov}}

\newcommand{\F}{\mathscr F}

\newcommand{\N}{\mathbb N}

\newcommand{\E}{\mathbb E}

\newcommand{\p}{\mathcal P_{\sim u}}

\newtheorem{example}{Example}[section]

\begin{document}

\maketitle
\newcommand{\slugmaster}{
\slugger{juq}{xxxx}{xx}{x}{x--x}}

\renewcommand{\thefootnote}{\fnsymbol{footnote}}
\newcommand{\argmin}{\operatornamewithlimits{argmin}}

\begin{abstract}
When performing global sensitivity analysis (GSA), it is often assumed, for the sake of simplicity, for lack of information, or for sheer expediency, that uncertain variables in the model are independent.  It is intuitively clear--and easily confirmed through simple examples--that applying a GSA method designed for independent variables to a set of correlated variables generally leads to results that hard to interpret, at best. We generalize the probabilistic framework for GSA pioneered  by Sobol' to  problems with correlated variables; this is done by reformulating his indices in terms of approximation errors rather than variance analysis. The implementation of the approach and its computational complexity are discussed and illustrated on synthetic examples. 
\end{abstract}



\pagestyle{myheadings}
\thispagestyle{plain}
\markboth{J.L. HART  AND P.A.~GREMAUD}{SOBOL' INDICES}

\section{Introduction}
\label{intro}
Let $f:\Omega \to \R$, $\Omega \subset \R^p$, be a function or model and  let $\x=(x_1,x_2,\dots,x_p) \in \Omega$ be the input variables of that model. Global sensitivity analysis (GSA) aims to quantify the relative importance of the variables $x_1, \dots, x_p$ in determining  $f$ \cite{saltellibook}. Such quantification  is a crucial step in the development of predictive models. 
Sobol' \cite{sobol93,sobol}  introduced the idea of measuring the importance of a variable or group of variables by (i) regarding them as random variables and (ii)  apportioning to each group its relative contribution to the total variance of $f(\x)$. The indices now carrying his name  formalize this probabilistic approach and are a recognized tool in  GSA. 

In that spirit, 
let $u=\{i_1,i_2,\dots,i_k\}$ be a  subset of $\{1,2,\dots ,p\}$ and let $\sim u =\{1,2,\dots,p\}\setminus u$ be its complement. We refer to the group of variables corresponding to $u$ as  $\x_u=(x_{i_1},x_{i_2},\dots,x_{i_k})$. 
Assuming $f$ to be  square integrable against a probability density function (PDF) $\phi$, we consider the decomposition
\begin{eqnarray}
f(\x)=f_0+\sum_{k=1}^p \sum_{|u|=k} f_u(\x_u). \label{anova}
\end{eqnarray}
where the $f_u$'s are defined recursively 
\begin{align}
\label{f_comp}
&f_0=\E[f(\x)],\\
&f_i(x_i)=\E[f(\x)|x_i]-f_0, \nonumber \\
&f_{i,j}(x_i,x_j)=\E[f(\x)|x_i,x_j]-f_i(x_i)-f_j(x_j)-f_0, \nonumber \\
\vdots \nonumber \\
&f_u(\x_u) = \E[f(\x)|\x_u]-\sum_{\substack{v \in P\\ v \subset u}} f_v(\x_v), \nonumber
\end{align}
where $P$ is the power set of $\{1,2,\dots,p\}$. In what follows, we omit writing $v \in P$ for such sums.
The decomposition \eqref{anova} is referred to as the ANOVA (analysis of variance) decomposition of $f(\x)$ if the input variables are {\em independent}; further, in that case, the $f_u$'s  satisfy
\begin{eqnarray}
\int f_u(\x_u)\phi(\x) dx_\ell=0 \quad \mbox{ for }\ell \in u. \label{anovaproperty}
\end{eqnarray}
If follows from \eqref{anovaproperty} that the $f_u(\x_u)$'s have mean zero and are mutually orthogonal. Consequently, we may decompose the variance of $f(\x)$ as
\begin{eqnarray}
\V(f(\x))=\sum_{k=1}^p \sum_{|u|=k} \V(f_u(\x_u)). \label{vardecomp}
\end{eqnarray}
This suggests the now classical definition of the Sobol' index (for independent variables) corresponding to  $\x_u$ 
\begin{eqnarray}
S_u^I=\frac{\V(f_u(\x_u))}{\V(f(\x))}.
\label{indsobol}
\end{eqnarray}
With independent variables, the Sobol' indices satisfy two key properties
\begin{align}
& \mbox{\bf conservation:} & \sum_{k=1}^p \sum_{|u|=k} S_u^I=1, \label{cp} \\
& \mbox{\bf boundedness:}  & \forall \,u, 0 \le S_u^I \le 1, \label{np}
\end{align}
where both properties easily follow from \eqref{vardecomp} and \eqref{indsobol}. The Sobol' indices with independent variables can be efficiently and accurately computed through the use of  ANOVA and Monte Carlo integration \cite{saltellialgorithm}.

When the input variables are {\em dependent}, the decomposition \eqref{anova} can still be considered but neither \eqref{anovaproperty} nor the resulting orthogonality properties hold. The Sobol' indices may then again be defined from \eqref{anova} \cite{li}
\begin{eqnarray}
S_u=\frac{\C(f_u(\x_u),f(\x))}{\V(f(\x))}, \label{covsobol}
\end{eqnarray}
where  \eqref{covsobol} clearly reverts to \eqref{indsobol} if the input variables are independent.
With dependent variables, the indices \eqref{covsobol} satisfy the conservation property but not the boundedness  property: while they sum up to 1, some of the  $S_u$'s may be negative and thus do not yield a quantitative measure of variable importance. We borrow  a simple example from \cite{staum} to illustrate this point. 

\begin{example}
\label{synexample}
Let 
\begin{eqnarray*}
f(x_1,x_2)=x_1+x_2
\end{eqnarray*}
where $x_1,x_2$ have a joint normal distribution with $\E[x_1]=\E[x_2]=0$, $\V(x_1)=\V(x_2)=1$, and $\C(x_1,x_2)=\rho \in (0,1)$. Then, following \eqref{anova},  $f$ admits the decomposition
\begin{eqnarray*}
f(x_1,x_2)=0+(1+\rho)x_1+(1+\rho)x_2+(-\rho)(x_1+x_2)
\end{eqnarray*}
and the associated Sobol' indices are $S_1=S_2=\frac{1+\rho}{2}$, $S_{1,2}=-\rho$. These indices sum to 1 but their interpretation as contributions to the variance is lost. 
\end{example}

For any $u \subset \{1,\dots,p\}$, the total Sobol' index $T_u$ is defined as the sum of all indices $S_v$ with $v \cap u \ne \emptyset$, i.e.
\begin{eqnarray}
T_u=\sum_{v \cap u \ne \emptyset} S_v. \label{covtotal}
\end{eqnarray}
If $x_1,x_2,\dots,x_p$ are independent then $S_u \le T_u$ so we may interpret $S_u$ as the exclusive contribution of $\x_u$ to $\V(f(\x))$ and $T_u$ as the contribution of $\x_u$ to $\V(f(\x))$ including its interactions with $\x_{\sim u}$. This does not generalize when $x_1,x_2,\dots,x_p$ possess dependences; it is possible to have $S_u > T_u$ so the interactions of $\x_u$ with $\x_{\sim u}$ give ``negative contributions," again a difficult concept to interpret. 

As highlighted in \cite{kucherenko}, an equivalent definition of the total Sobol' index, with independent or dependent variables, is given by
\begin{eqnarray*}
T_u = 1-\frac{\V(\E[f(\x) \vert \x_{\sim u}])}{\V(f(\x))} = \frac{\E[\V(f(\x) \vert \x_{\sim u})]}{\V(f(\x))}.
\end{eqnarray*}

The issue of GSA with dependent variables has been the object of intense recent research. Regarding Sobol' indices,  Xu and Gertner \cite{xugertner} propose a decomposition of the Sobol' indices into a correlated and uncorrelated part for linear models. Li et al. \cite{li} build upon this to decompose the Sobol' indices for a general model. Mara and Tarantola \cite{mara} propose to use the Gram-Schmidt process to decorrelate the inputs variables then define new indices through the Sobol' indices of the decorrelated problem. Building off the older work of \cite{stone} and \cite{hooker}, Chastaing et al. \cite{chastaing} provide a theoretical framework to generalize the ANOVA decomposition to problems with dependent variables. In a subsequent article \cite{chastaingalgorithm}, they also provide a computational algorithm to accompany their theoretical work. In contrast to the other works which focus on generalizing the ANOVA decomposition, Kucherenko et al. \cite{kucherenko} develop the Sobol' indices via the law of total variance. Recent work of Mara and Tarantola \cite{fast_dependent_variables} considered estimating Sobol' indices with dependent variables using the Fourier Amplitude Sensitivity Test.

Challenges with Sobol' indices, as highlighted in the works cited above, have motivated interest in other approaches. Staum et al. \cite{staum} suggest Shapley effects \cite{owen2014} as an alternative tool; \cite{iooss_prieur_shapley,owen_prieur_shapley} have since extended this work. Borgonovo \cite{borgonovo2,borgonovo1,miim_uq_handbook} develops a global sensitivity measure using conditional density functions instead of conditional expectations. Building upon this, Pianosi and Wagener \cite{cdfmethod} propose a measure of global sensitivity using the conditional distribution functions. All of these approaches suffer from interpretability and/or computability issues at some level.

Risk analysis, reliability engineering, variable prioritization for data acquisition, model development/analysis, and dimension reduction are a few of the possible applications of GSA. To facilitate our analysis, this article focuses on applying Sobol' indices for dimension reduction, i.e. approximating $f(\x)$ by a function which depends on fewer variables. To this end, we characterize the total Sobol' index, with independent or dependent variables, in terms of approximation error rather than variance analysis, by developing the approximation theoretic analogue of the law of total variance approach in \cite{kucherenko}. Section~\ref{definitionproperties} details this approximation theoretic characterization. We then use this characterization in Section~\ref{sec:dim_red} to analyze the error created through fixing unimportant variables, a common approach to dimension reduction. Practical and computational considerations are highlighted in Section~\ref{sec:practical_considerations}. Section~\ref{sec:examples} provides examples to illustrate the properties discussed in the previous sections. We conclude with forward looking remarks in Section~\ref{sec:conclusion}. 

\section{An Approximation Theoretic Perspective of Sobol' Indices}
\label{definitionproperties}

For $k=1,\dots,p$,   let $\Omega_k$ be the  set in which the input variable $x_k$ takes its values. Setting $\Omega=\Omega_1\times \Omega_2 \times \cdots \times \Omega_p$, we consider the probability space  $(\Omega,\F,\mu)$ as well as the Hilbert space $L^2(\Omega)$ of square integrable functions on this probability space. 

Let $f:\Omega \to \R$ be square integrable and $f_0=\E[f(\x)]$. Given $u \subset \{1,2,\dots,p\}$, we ask ``how accurately can $f(\x)-f_0$ can be approximated without the variables $\x_u$?" In other words, what is the error associated with the approximation
\begin{eqnarray}
f(\x) - f_0 \approx \p f(\x_{\sim u}), \label{approx}
\end{eqnarray}
where  $\p f(\x_{\sim u})$ is the optimal $L^2(\Omega)$ approximation of $f(\x)-f_0$ which does not depend on $\x_u$? We show that this error,
\begin{eqnarray}
\frac{\|(f(\x)-f_0)-\p f(\x_{\sim u}) \|_2^2}{\|f(\x)-f_0\|_2^2}, \label{generalizedindex}
\end{eqnarray}
coincides with the classical definition of the total Sobol' index \eqref{covtotal}. This requires a few technical considerations. 

For $f\in L^2(\Omega)$, we say that {\em $f$ does not depend on $x_k$} if and only if
\begin{eqnarray}
\mbox{there exists } N \in \F \mbox{ with } \mu(N) = 0 \mbox{ such that }
f(\x)=f(\y)\quad \forall \x,\y\in \Omega \setminus N \mbox{ with } \x_{\sim k}=\y_{\sim k}.
\label{f_dependend}
\end{eqnarray}
Otherwise, $f$ is said to depend on $x_k$. For any subset $v$ of $\{1,\dots, p\}$, we define 
\begin{eqnarray*}
M_v = \left\{ f \in L^2(\Omega) \vert f \mbox{ satisfies } \eqref{f_dependend} \mbox{ } \forall k \in \sim v \right\}
\end{eqnarray*}
as the set of all  functions in  $L^2(\Omega)$  that do not depend on any variables in $\x_{\sim v}$. Roughly speaking, $M_v$ is the set of those functions of $L^2(\Omega)$ that depend on $\x_v$. We prove in the Appendix (Proposition~\ref{prop:csubspace}) that $M_v$ is a closed subspace of $L^2(\Omega)$; consequently,
 $L^2(\Omega)$ can be decomposed as a direct sum of $M_v$ and $M_v^\perp$, the orthogonal complement of $M_v$, i.e.
\begin{eqnarray}
L^2(\Omega) =M_v \oplus M_v^\perp. \label{directsum}
\end{eqnarray}
It is worth noting here that $M_v^\perp \ne M_{\sim v}$. 

Setting $v = \sim u$, we can now rewrite  \eqref{approx} more explicitly as
\begin{eqnarray}
f(\x) = f_0 + \p f(\x_{\sim u})+\p^\perp f(\x), \label{ortho}
\end{eqnarray}
where  $\p f(\x_{\sim u})=\E[f(\x)-f_0|\x_{\sim u}]=\E[f(\x)|\x_{\sim u}]-f_0$ is the projection of $f(\x)-f_0$ onto $M_{\sim u}$. This orthogonal decomposition yields 
\begin{eqnarray}
\frac{\|(f(\x)-f_0)-\p f(\x_{\sim u}) \|_2^2}{\|f(\x)-f_0\|_2^2}=\frac{\| \p^\perp f(\x)\|^2}{\|f(\x)-f_0\|_2^2}=1-\frac{\| \p f(\x_{\sim u})\|_2^2}{\|f(\x)-f_0\|_2^2}. \label{generalizedindex2}
\end{eqnarray}

Proposition~\ref{prop:Tk} shows that the total Sobol' index $T_u$ equals \eqref{generalizedindex}, thus providing a new characterization of the total Sobol' indices, with independent or dependent variables, and giving a clear interpretation of  these indices in terms of relative approximation error. We note that the decomposition \eqref{ortho} is the approximation theoretic analogue of the law of total variance approach in \cite{kucherenko}.

\begin{proposition}
For $u \subset \{1,2,\dots,p\}$,
\begin{eqnarray*}
T_u = \frac{\|(f(\x)-f_0)-\p f(\x_{\sim u})\|_2^2}{\|f(\x)-f_0\|_2^2}.
\end{eqnarray*}
\label{prop:Tk}
\end{proposition}

\begin{proof}
By rearranging  \eqref{anova}, we get 
\begin{eqnarray*}
f(\x)-f_0=\sum_{v \cap u=\emptyset} f_v(\x_v)+\sum_{v \cap u \ne \emptyset} f_v(\x_v).
\end{eqnarray*}
Using \eqref{f_comp} and the fact that $\p f(\x_{\sim u})$ is the $L^2(\Omega)$ projection of $f(\x)-f_0$ onto $M_{\sim u}$, we observe
\begin{align*}
\p f(\x_{\sim u}) & = \E[f(\x)-f_0|\x_{\sim u}] \\
& =f_{\sim u}(\x_{\sim u}) + \sum_{v \subset \sim u} f_v(\x_v) \\
& =\sum_{v \cap u=\emptyset} f_v(\x_v).
\end{align*}
Taking into account \eqref{covtotal}, \eqref{ortho}, \eqref{generalizedindex2}, and the linearity of the covariance operator, it follows that 
\begin{align*}
T_u = & \sum_{v \cap u \ne \emptyset} \frac{\C(f_v(\x_v),f(\x))}{\V(f(\x))}\\
=&\frac{1}{\|f(\x)-f_0\|_2^2} \C\left(\sum_{v \cap u \ne \emptyset} f_v(\x_v),\p f(\x_{\sim u})+\p^\perp f(\x) \right) \\
= & \frac{\C(\p^\perp f(\x),\p f(\x_{\sim u}))}{\|f(\x)-f_0\|_2^2} + \frac{\C(\p^\perp f(\x),\p^\perp f(\x))}{\|f(\x)-f_0\|_2^2} \\
=& \frac{0}{\|f(\x)-f_0\|_2^2} + \frac{ \|\p^\perp f(\x) \|_2^2}{\|f(\x)-f_0\|_2^2 }\\
=&\frac{\|(f(\x)-f_0)-\p f(\x_{\sim u})\|_2^2}{\|f(\x)-f_0\|_2^2}
\end{align*}
\end{proof}

Because of their approximation theoretic interpretation, this article focuses on the total Sobol' indices. For simplicity we will omit ``total" in what follows and refer to $T_u$ as the Sobol' index.

\section{Applying the Approximation Theoretic Perspective for Dimension Reduction}
\label{sec:dim_red}

One common use of the Sobol' indices is dimension reduction, i.e., approximating $f$ by a function which depends on fewer variables. There are several ways to do this, three examples are:
\begin{enumerate}
\item projecting $f$ onto a subspace of functions which only depend on a subset of the input variables,
\item constructing a surrogate model using only a subset of the variables, for instance, training a statistical model (such as a Gaussian process) with only a subset of the input variables
\item fixing some of the input variables to nominal values, or possibly a function of the other input variables.
\end{enumerate}
The approximation theoretic perspective in Section~\ref{definitionproperties} provides useful insights for all three of these possible approaches. For the first approach, the total Sobol' index $T_u$ is the relative $L^2(\Omega)$ error squared when $f$ is approximated by its orthogonal $L^2(\Omega)$ projection onto the subspace of functions which only dependent on $\x_{\sim u}$. However, acquiring this projection is computationally costly since it requires computing many high dimensional integrals, so this approach is limited in practice. The second approach is practical in cases where the user wishes to use existing evaluations of $f$ to train a surrogate model. The unimportant variables by the considered as latent and the surrogate model may be trained using only a subset of input variables. Since $T_u$ is the error for the optimal $L^2(\Omega)$ approximation, it provides a lower bound on the $L^2(\Omega)$ error of a surrogate model approximation. Hence $T_u$ is useful for making decisions about which variables to use when constructing a surrogate model in the second approach. The third approach, fixing inputs, is commonly used in practice because of its simplicity. As demonstrated below, the approximation theoretic perspective of Sobol' indices is useful for analyzing approximation error in this setting as well.

With the assumption of independent variables, classical results exist \cite{freezing_variables} which use the Sobol' indices to bound the error incurred when fixing variables to nominal values. Partition $\x=(\x_u,\x_{\sim u})$ and assume that $T_u$ is small. We would like to approximate $f$ by replacing $\x_u$ with a function of $\x_{\sim u}$. Specifically, we approximate $f(\x)$ by $f(g(\x_{\sim u}),\x_{\sim u})$ where $g(\x_{\sim u})$ is an approximation of $\x_u$. It is common to take the constant approximation $g(\x_{\sim u})=\E[\x_u]$ when the variables are independent. Our subsequent analysis considers a general $g$. 

The relative error incurred by replacing $\x_u$ with $g(\x_{\sim u})$ is
\begin{align}
\label{delta}
\delta_u = & \frac{\vert \vert f(\x) - f(g(\x_{\sim u}),\x_{\sim u}) \vert \vert_2^2}{\vert \vert f(\x)-f_0 \vert \vert_2^2} .
\end{align}

Proposition~\ref{prop:lower_bound} extends a result in \cite{freezing_variables} to the case with dependent variables.

\begin{proposition}
\label{prop:lower_bound}
For any $u \subset \{1,2,\dots,p\}$ and any $g:\Omega_{\sim u} \to \Omega_u$ such that $f(g(\x_{\sim u}),\x_{\sim u}) \in L^2(\Omega),$
\begin{eqnarray*} 
\delta_u \ge T_u.
\end{eqnarray*}
\end{proposition}
\begin{proof}
The result follows since $T_u$ is the squared relative $L^2(\Omega)$ error of the the orthogonal projection of $f(\x)-f_0$ onto $M_{\sim u}$, i.e. the optimal approximation in $M_{\sim u}$, and $f(g(\x_{\sim u}),\x_{\sim u}) \in M_{\sim u}$.
\end{proof}

An upper bound on $\delta_u$ is more useful than a lower bound in most cases; however, a tight upper bound is difficult to attain. Plugging \eqref{ortho} into \eqref{delta} yields
\begin{eqnarray}
\label{delta_proj}
\delta_u = \frac{\vert \vert \p^\perp f(\x) - \p^\perp f(g(\x_{\sim u}),\x_{\sim u}) \vert \vert_2^2}{\vert \vert f(\x)-f_0 \vert \vert_2^2} .
\end{eqnarray}

Recall that the Sobol' index $T_u$, which we assume to be small, is given by
\begin{eqnarray*}
T_u = \frac{ \vert \vert \p^\perp f(\x) \vert \vert_2^2}{ \vert \vert f(\x)-f_0 \vert \vert_2^2} .
\end{eqnarray*}
Hence, $\p^\perp f(\x)$ is small relative to $f(\x)-f_0$.

Proposition~\ref{prop:upper_bound} provides a loose, but informative, upper bound on $\delta_u$.

\begin{proposition}
For any $u \subset \{1,2,\dots,p\}$ and any $g:\Omega_{\sim u} \to \Omega_u$ such that $f(g(\x_{\sim u}),\x_{\sim u}) \in L^2(\Omega),$
\label{prop:upper_bound}
\begin{eqnarray*}
\delta_u \le T_u + \frac{\vert \vert \p^\perp f(g(\x_{\sim u}),\x_{\sim u}) \vert \vert_2^2}{\vert \vert f(\x)-f_0 \vert \vert_2^2} + 2 T_u \frac{\vert \vert \p^\perp f(g(\x_{\sim u}),\x_{\sim u}) \vert \vert_2}{\vert \vert \p^\perp f(\x) \vert \vert_2}
\end{eqnarray*}
\end{proposition}

\begin{proof}
Notice,
\begin{align*}
\vert \vert \p^\perp f(\x) - \p^\perp f(g(\x_{\sim u}),\x_{\sim u}) \vert \vert_2^2  = & \vert \vert \p^\perp f(\x) \vert \vert_2^2 \\ & + \vert \vert \p^\perp f(g(\x_{\sim u}),\x_{\sim u}) \vert \vert_2^2 \\
& - 2 \E [ \p^\perp f(\x) \p^\perp f(g(\x_{\sim u}),\x_{\sim u}) ] .
\end{align*}
Applying the Triangle inequality and Cauchy-Schwarz inequality we have
\begin{align*}
\vert \vert \p^\perp f(\x) - \p^\perp f(g(\x_{\sim u}),\x_{\sim u}) \vert \vert_2^2  \le & \vert \vert \p^\perp f(\x) \vert \vert_2^2 \\ & + \vert \vert \p^\perp f(g(\x_{\sim u}),\x_{\sim u}) \vert \vert_2^2 \\
& + 2  \vert \vert \p^\perp f(\x) \vert \vert_2 \vert \vert \p^\perp f(g(\x_{\sim u}),\x_{\sim u}) \vert \vert_2 .
\end{align*}
Multiplying and dividing $2  \vert \vert \p^\perp f(\x) \vert \vert_2 \vert \vert \p^\perp f(g(\x_{\sim u}),\x_{\sim u}) \vert \vert_2$ by $\vert \vert \p^\perp f(\x) \vert \vert_2$ and dividing both sides of the inequality by $\vert \vert f(\x)-f_0\vert \vert_2^2$ completes the proof.
\end{proof}

Observe that if 
\begin{eqnarray*}
\vert \vert \p^\perp f(\x) \vert \vert_2 = \vert \vert  \p^\perp f(g(\x_{\sim u}),\x_{\sim u}) \vert \vert_2
\end{eqnarray*}
 then 
\begin{eqnarray*}
T_u \le \delta_u \le 4 T_u .
\end{eqnarray*}
This assumption typically does not hold in practice, nor is it verifiable; however, it provides some intuition about the behavior of the error. In particular, $\delta_u$ will be small when $\vert \vert  \p^\perp f(g(\x_{\sim u}),\x_{\sim u}) \vert \vert_2$ is approximately $\vert \vert \p^\perp f(\x) \vert \vert_2$. The error, $\delta_u$, will be large when the magnitude of $\p^\perp f(\x)$ increases dramatically on subsets of $\Omega$ which have a small probability under $\x$ and a larger probability under $(g(\x_{\sim u}),\x_{\sim u})$. The magnitude of $\delta_u$ is closely linked to how well the distribution of $(g(\x_{\sim u}),\x_{\sim u})$ approximates the distribution of $\x$, and the robustness of the Sobol' index which respect to changes in the distribution of $\x$, i.e. how much $T_u$ changes when the distribution of $\x$ is changed. An algorithm for testing such robustness is given in \cite{hart_robustness}.

Three conclusions may be drawn from the arguments above:
\begin{enumerate}
\item Dependencies between the variables can help reduce $\delta_u$.
\item A tight upper bound will be difficult attain without placing additional assumptions on the behavior of $f$ on sets of small probability.
\item Testing the robustness of $T_u$ with respect to changes in the distribution of $\x$ provides a heuristic to asses when $\delta_u$ will be small.
\end{enumerate}
 
\section{Practical and Computational Considerations}
\label{sec:practical_considerations}

The Sobol' indices may be estimated via Monte Carlo integration \cite{kucherenko} or the Fourier Amplitude Sensitivity Test \cite{fast_dependent_variables}. In what follows, $\{T_k\}_{k=1}^p$ is estimated via Monte Carlo integration using $(p+1)N$ evaluations of $f$, where $N$ is the number of Monte Carlo samples \cite{kucherenko}.

When the variables are independent we have,
\begin{eqnarray}
\min_{k \in u} T_k \le T_u \le \sum_{k \in u} T_k ,
\label{T_u_inference}
\end{eqnarray}
for any $u \subset \{1,2,\dots,p\}$. In this case it is typically sufficient to compute $\{T_k\}_{k=1}^p$ as inferences about $T_u$ may be made using $\{T_k\}_{k=1}^p$ with \eqref{T_u_inference}. For instance, if $ \sum_{k \in u} T_k$ is small then we know that $T_u$ is small. This does not generalize when the variables are dependent. The example in Subsection~\ref{corr_knob} provides a case where $T_1$ and $T_2$ are small, but $T_{1,2}$ is large. The approximation theoretic framework is helpful for interpreting this. When $f$ is sensitive to two variables which are dependent on one another then one variable may be projected out with little error because the remaining variable can approximate its influence on $f$; however, when both are projected out a large error is incurred.

A practical strategy with dependent variables is to estimate $\{T_k\}_{k=1}^p$, which requires $(p+1)N$ evaluations of $f$. Then $\{T_k\}_{k=1}^p$ may be analyzed, along with information about the dependencies in $\x$ (known analytically or from the samples), and the user may select particular subsets $u \subset \{1,2,\dots,p\}$ for which to compute $T_u$. Using the estimator from \cite{kucherenko}, the additional cost to compute $T_u$ for a given subset $u$ will be $N$ evaluations of $f$.

The robustness of $T_k$ to changes in the distribution of $\x$ may be computed as a by-product of computing $\{T_k\}_{k=1}^p$ \cite{hart_robustness}. If $T_k$, $k \in u$, is not robust to changes in the distribution of $\x$, then $\delta_u$ may be significantly larger than $T_u$. For a particular $g$ and subset $u$, the user may compute $\delta_u$ directly; this also requires $N$ evaluations of $f$.

Choosing $g$ is a challenge in practice. The natural choice, $g(\x_{\sim u}) = \E[\x_u]$, fails to exploit dependency information and is not suggested. Rather, we suggest $g(\x_{\sim u})=\E[\x_u \vert \x_{\sim u}]$ since (i) linear dependencies are common in practice (normal distributions and copula models are two common examples), and (ii)$ \E[\x_u \vert \x_{\sim u}]$ is easily computed (either analytically or through linear regression with the existing samples). If the dependencies in $\x$ are known to be nonlinear then $g$ may be estimated by nonlinear regression (using the existing samples). The challenge in this case is determining an appropriate nonlinear model for $g$.

\section{Illustrative Examples}
\label{sec:examples}
This section provides two illustrative examples to highlight properties of the Sobol' indices and their association with approximation error. 
\subsection{A Linear Function}
\label{corr_knob}
Let 
\begin{eqnarray}
f(\x)=20x_1+16x_2+12x_3+10x_4 + 4x_5 \label{simpleex}
\end{eqnarray}
and $\x$ follow a multivariate normal distribution with mean $\mu$ and covariance matrix $\Sigma$ given by
\[ \mathbf{\mu}= \left[ \begin{array}{cc}
0 \\
0 \\
0 \\
0 \\
0\\
\end{array} \right],\hspace{10 mm}
\Sigma= \left[ \begin{array}{ccccc}
1 &.5\rho & .5 \rho & 0 & .8 \rho \\
.5\rho & 1 & 0 & 0 & 0 \\
.5 \rho & 0 & 1 & 0 & .3 \rho \\
0 & 0 & 0 & 1 & 0\\
.8 \rho & 0 & .3\rho & 0 & 1 \\
\end{array} \right], \qquad 0\le \rho\le 1. \] 

The Sobol' indices $T_k$, $k=1,\dots,5$, are computed analytically and displayed in Figure~\ref{fig:corrknob} as a function of $\rho$. Observe that the ordering of importance changes as  the correlations become stronger. This underscores the significance of accounting for dependencies. Also notice that the Sobol' indices are decreasing as a function of $\rho$. The approximation theoretic perspective provides a nice interpretation of this. As the correlations are strengthened, the error associated with projecting out a variable decreases because its influence on $f(\x)$ may be approximated by the other variables.

\begin{figure}[h]
\centering
\includegraphics[width=.55 \textwidth]{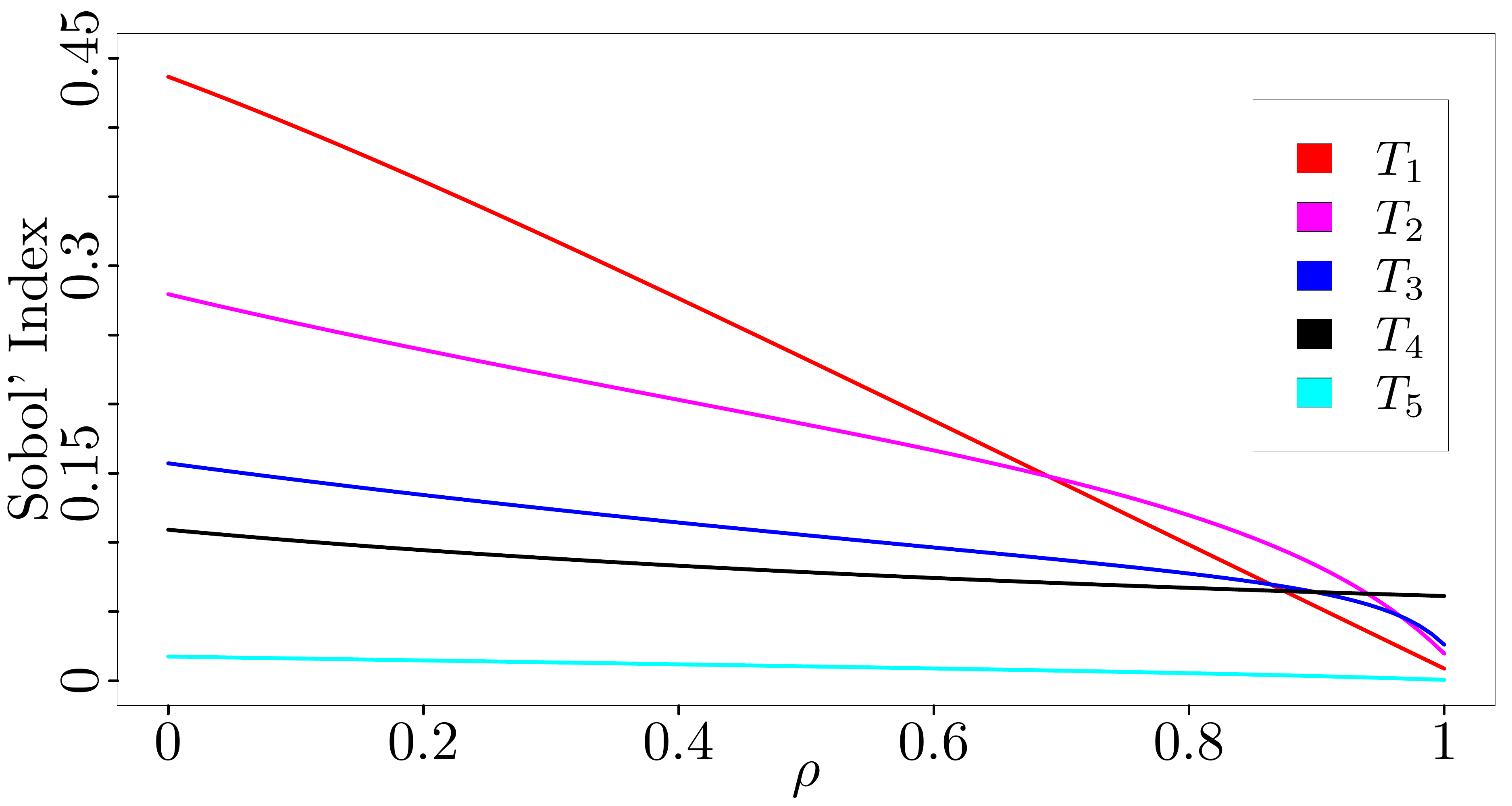}
\caption{Sobol' indices for \eqref{simpleex} with increasing correlation strength as $\rho$ varies from 0 to 1.}
\label{fig:corrknob}
\end{figure}

Table~\ref{tab:totalindices} displays the Sobol' indices $T_1,T_2$, and $T_{1,2}$ when $\rho=1$. This demonstrates that two variables may have small Sobol' indices individually ($T_1$ and $T_2$), but their joint Sobol' index ($T_{1,2}$) may be significantly larger. This phenomenon, which does not occur when the variables are independent, is important when analyzing results with dependent variables; inference about subsets cannot be made with $\{T_k\}_{k=1}^p$ alone.

\begin{table}[h]
\centering
\ra{1.5}
\begin{tabular}{ccc}
\toprule
$T_1$ & $T_2$ & $T_{1,2}$\\
0.0087 & 0.0196 & 0.4228\\
\bottomrule
\end{tabular}
\vskip.5cm
\caption{Sobol' indices of \eqref{simpleex} for variables $x_1$, $x_2$, and $(x_1,x_2)$ when $\rho = 1$.}
\label{tab:totalindices}
\end{table}

\subsection{A Nonlinear Function}
\label{gfunsection}
Let $f$ be the g-function of \cite{sobol} with $p=10$ variables; more precisely, $f$ is given by
\begin{eqnarray}
f(\x)=\prod_{k=1}^{10} \frac{\vert 4x_k-2 \vert +a_k}{1+a_k},
\label{g_fun}
\end{eqnarray}
where the parameters $a_k$, $k=1,2,\dots,10$, is given by $\mathbf a = (1,2,3,9,11,13,20,25,30,35)$. Let $\x$ follow a multivariate normal distribution with mean $\mu \in \R^{10}$, 
$$\mu_k = \frac{1}{2}, \qquad k=1,2,\dots,10,$$
and covariance matrix $\Sigma \in \R^{10 \times 10}$,
$$\Sigma_{k,k}=\frac{1}{6}, \qquad k=1,2,\dots,10, \qquad \text{ and } \qquad \Sigma_{i,j} = \frac{\rho}{6 \vert i-j+1 \vert^\frac{1}{\gamma}}, \qquad i\ne j .$$

The covariance matrix is parameterized so that the magnitude of the covariances are large near the diagonal of $\Sigma$ and decrease as they move away from the diagonal. The parameter $\gamma$ determines the rate at which they decrease, as $\gamma \to \infty$, the off diagonal elements of $\Sigma$ all converge to $\rho/6$. Hence $\gamma$ tunes how many variables are strongly correlated with one another. The parameter $\rho$ scales the strength of the correlations.

Direct calculations yield that variables $x_i$, $i=7,8,9,10$, are not influential for any $\rho$, $\gamma$, though $T_{7,8,9,10}$ does depend on $\rho$ and $\gamma$. Figure~\ref{fig:gfun_gaussian} demonstrates how the Sobol' index $T_{7,8,9,10}$ and the approximation error $\delta_{7,8,9,10}$ vary with respect to $\rho$ and $\gamma$. On the left panel we fix $\gamma = 1$ and vary $\rho$ from 0 to 1; on the right panel we fix $\gamma = 6$ and vary $\rho$ from 0 to 1.

\begin{figure}[h]
\centering
\includegraphics[width=.49 \textwidth]{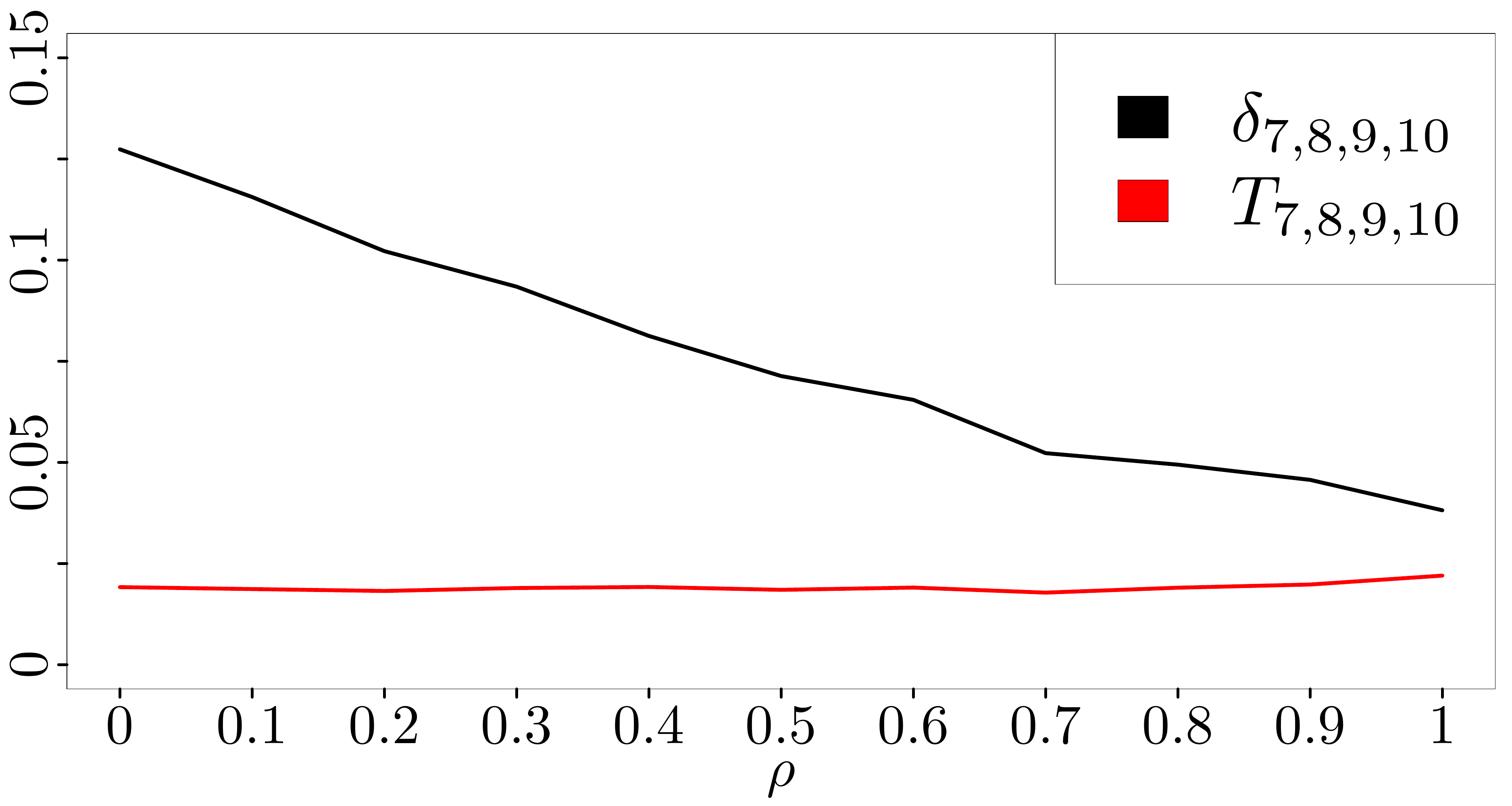}
\includegraphics[width=.49 \textwidth]{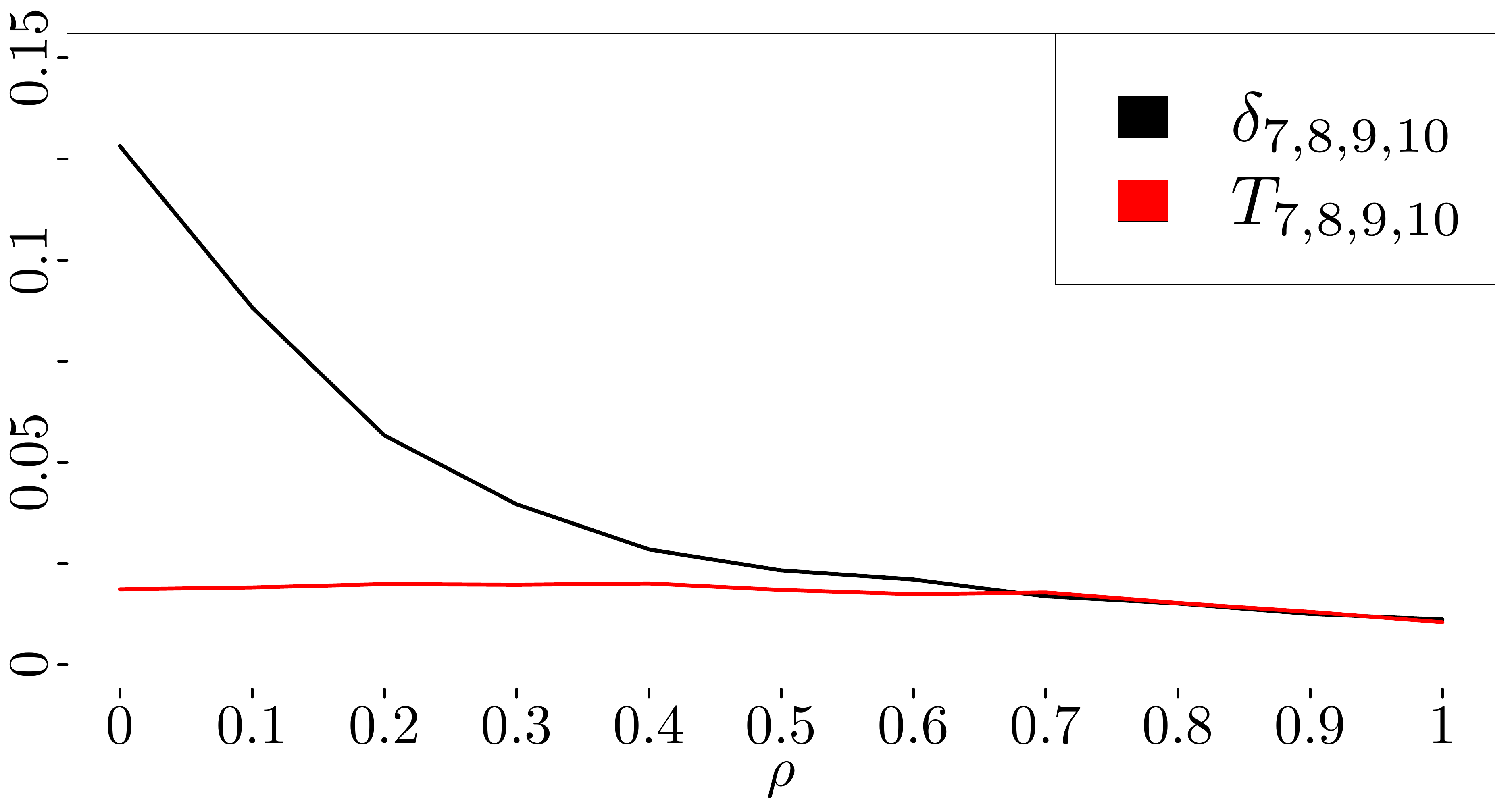}
\caption{Sobol' index $T_{7,8,9,10}$ and approximation error $\delta_{7,8,9,10}$ for \eqref{g_fun} as $\rho$ varies from 0 to 1. Left: $\gamma=1$; right: $\gamma=6$.}
\label{fig:gfun_gaussian}
\end{figure}

Figure~\ref{fig:gfun_gaussian} shows that linear dependencies, as in the case of a multivariable normal random vector, aid in approximating $f$ by fixing unimportant variables. In particular, we observe that the error from replacing $\x_{7,8,9,10}$ with its conditional expectation decreases as $\rho$ increases. Taking a larger $\gamma$, as in the right panel, corresponds to having more variables which are strongly correlated. Having $T_{7,8,9,10} \approx \delta_{7,8,9,10}$, as in the right panel with $\rho \ge 0.7$, demonstrates that the optimal approximation may be attained by replacing $\x_{7,8,9,10}$ with its conditional expectation.

\section{Conclusion}
\label{sec:conclusion}
This article provides a framework to analyze dimension reduction with dependent variables and highlights how dependencies may aid the user in dimension reduction. The approximation theoretic characterization of Sobol' indices is useful as it demonstrates how Sobol' indices are linked to optimal approximation and how they may be used to analyze the error when replacing variables $\x_u$ with a function $g(\x_{\sim u})$, a common approach in practice. An important factor in this analysis is the robustness of the Sobol' indices to changes in the distribution of $\x$ \cite{hart_robustness}. Further analysis is needed to (i) connect robustness studies to dimension reduction and (ii) determine an optimal $g$, particularly in the presence of nonlinear dependencies.

There has also been recent progress with derivative-based global sensitivity indices \cite{dgsm_poincare,dgsm1,dgsm2} and active subspaces \cite{active_subspaces} as alternative approaches for dimension reduction with independent variables. Extending analysis of these methods for dimension reduction with dependent variables is another avenue of future research. 

The approximation theoretic perspective provides a useful characterization of Sobol' indices for dimension reduction with dependent variables. However, this characterization does not address difficulties which arise in other applications of GSA. For instance, when GSA is used for model development the user wants to identify the most important variables. If the most important variables have strong dependencies then their Sobol' indices may be small which hides the information the user desires. Future work may consider alternative characterizations which are focused toward other applications. This may involve the Sobol' indices, or possibly other tools in GSA. 
%
%
%
%
\appendix
\section*{Appendix}

\begin{proposition}
$M_u$ is a closed subspace of $L^2(\Omega)$.
\label{prop:csubspace}
\end{proposition}
\begin{proof}
$M_u$ is clearly a subset of $L^2(\Omega)$. To show that it is closed, 
let $\{f_n\}$ be a sequence in $M_u$ which converges to $f \in L^2(\Omega)$. We want to show that $f \in M_u$. Suppose by contradiction that $f \notin M_u$. Then $\exists i \in \{1,2,\dots,p\}$ such that $i \notin u$ and $f$ depends on $x_i$. Then $\exists A \in \F$ and $\x,\y \in A$ such that $\mu(A)>0$ with $\x_{\sim i}=\y_{\sim i}$ and $f(\x)\ne f(\y)$. Since $f_n \to f$ in $L^2(\Omega)$ then $\exists \{f_{n_k}\}$, a subsequence of $\{f_n\}$, such that $f_{n_k} \to f$ point wise almost everywhere. Since $\x, \y \in A$ and $\mu(A)>0$ then $f_{n_k}(\x) \to f(\x)$ and $f_{n_k}(\y) \to f(\y)$. But $f_{n_k}$ do not depend on $x_i$ so $f_{n_k}(\x)=f_{n_k}(\y)$ $\forall k \in \N \implies f(\x)=f(\y)$. This is a contradiction so $f \in M_u$ and hence $M_u$ is closed.
\end{proof}

\bibliographystyle{siam}

\bibliography{Dep_Var}
\end{document}